\documentclass[12pt, draft, onecolumn, letterpaper]{IEEEtran}

\usepackage{amsfonts, amssymb}
\usepackage{amscd}
\usepackage{amsmath}
\usepackage{comment}
\usepackage{pdflscape}
\usepackage{amsthm}
\usepackage{color}

 \newcommand{\Z}{\mathbb{Z}}
 \newcommand{\add}{{\mathbb{Z}_2\mathbb{Z}_4}}
 \newcommand{\C}{{\cal C}}
 \newcommand{\zero}{{\mathbf{0}}}
 \newcommand{\cG}{{\cal G}}
 \newcommand{\vv}{{\bf  v}}
 \newcommand{\vw}{{\bf  w}}

\newtheorem{theorem}{Theorem}
\newtheorem{definition}[theorem]{Definition}

\newtheorem{proposition}[theorem]{Proposition}
\newtheorem{lemma}[theorem]{Lemma}
\newtheorem{example}{Example}

\title{Binary Images of $\add$-Additive Cyclic Codes\thanks{%
  This work has been partially supported by the Spanish MINECO grants 
TIN2016-77918-P(AEI/FEDER, UE) and MTM2015-69138-REDT, and by the Catalan AGAUR 
grant 2014SGR-691.}}

\author{J.~Borges, S.T.~Dougherty,
C.~Fern\'{a}ndez-C\'{o}rdoba and R. Ten-Valls\thanks{S.T.~Dougherty is 
at 
Department of
Mathematics, University of Scranton,
Scranton, PA 18510, USA (email: doughertys1@scranton.edu).}\thanks{J.~Borges, 
C.~Fern{\'a}ndez-C{\'o}rdoba and R. Ten-Valls are members of the
Department of Information and Communications Engineering, Universitat
Aut\`{o}noma de Barcelona, 08193-Bellaterra,
Spain (email:~\{jborges, cfernandez, rten\}@deic.uab.cat).}}

\date{\today}

\begin{document}

\maketitle

\begin{abstract}
A ${\mathbb{Z}}_2{\mathbb{Z}}_4$-additive code ${\cal 
C}\subseteq{\mathbb{Z}}_2^\alpha\times{\mathbb{Z}}_4^\beta$ is called cyclic if 
the set of coordinates can be partitioned into two subsets, the set of 
${\mathbb{Z}}_2$ and the set of ${\mathbb{Z}}_4$ coordinates, such that any 
cyclic shift of the coordinates of both subsets leaves the code invariant. We 
study the binary images of $\Z_2\Z_4 $-additive cyclic codes. We determine all 
$\add$-additive cyclic codes with odd $\beta$ whose Gray images are linear 
binary codes.
\end{abstract}

\section{Introduction}

Denote by  $\Z_2$ and $\Z_4$ the rings of integers modulo 2 and modulo 4, 
respectively.  We denote the space of $n$-tuples over these rings as $\Z_2^n$ 
and $\Z_4^n$. A binary code is any non-empty subset $C$ of $\Z_2^n$. If that 
subset is a vector space then we say that it is a linear code.  Any non-empty 
subset $\C$ of $\Z_4^n$ is a quaternary code and a submodule of $\Z_4^n$ is 
called a linear code over $\Z_4$.

In 1994, it is proven by Hammons \textit{et al.}, \cite{sole}, that certain 
good 
non-linear binary codes can be seen as binary images of linear codes over 
$\Z_4$ 
under the Gray map. After \cite{sole}, the study of codes over $\Z_4$ and other 
finite rings has been developing and the construction of Gray maps has been a 
topic of study. In particular, the study of the structure of cyclic codes over 
$\Z_4$ was discussed by Calderbank \emph{et al.} \cite{Z4CyclicCalderbank} and 
by Pless and Qian \cite{PQ}.


In Delsarte's paper \cite{del}, he defines additive codes as  subgroups of the 
underlying abelian group in a translation association scheme. For the binary 
Hamming scheme, namely, when the underlying abelian group is of order $2^{n}$, 
the only structures for the abelian group are those of the form 
$\Z_2^\alpha\times \Z_4^\beta$, with $\alpha+2\beta=n$. This means that  the 
subgroups of $\Z_2^\alpha\times \Z_4^\beta$ are the only additive codes in a 
binary Hamming scheme, \cite{transProper}.

A $\mathbb{Z}_2\mathbb{Z}_4$-linear code is a binary image of a 
$\mathbb{Z}_2\mathbb{Z}_4$-additive code, that is an additive subgroup of 
$\mathbb{Z}_2^\alpha\times\mathbb{Z}_4^\beta$. These 
$\mathbb{Z}_2\mathbb{Z}_4$-linear codes were first introduced by Rif\`{a} and 
Pujol in 1997, \cite{transProper}, as abelian translation-invariant propelinear 
codes. Later, an exhaustive description of $\mathbb{Z}_2\mathbb{Z}_4$-linear 
codes was done by Borges \emph{et al.} in \cite{AddDual}. The structure and 
properties of $\mathbb{Z}_2\mathbb{Z}_4$-additive codes have been intensely 
studied, for example in \cite{MDS}, \cite{z2z4SD}, \cite{Z2Z4RK}. 

In \cite{Abu}, Abualrub \emph{et al.} define 
${\mathbb{Z}}_2{\mathbb{Z}}_4$-additive cyclic codes. A code in 
$\mathbb{Z}_2^\alpha\times\mathbb{Z}_4^\beta$ is called cyclic if the set of 
coordinates can be partitioned into two subsets, the set of coordinates over 
${\mathbb{Z}}_2$ and the set of coordinates over ${\mathbb{Z}}_4$, such that 
any 
cyclic shift of the coordinates of both subsets leaves the code invariant. 
These 
codes can be identified as submodules of the $\mathbb{Z}_4[x]$-module 
$\mathbb{Z}_2[x]/\langle x^\alpha-1\rangle\times\mathbb{Z}_4[x]/\langle 
x^\beta-1\rangle$. 

In \cite{sole}, the authors give necessary and sufficient conditions for a
binary code to be $\Z_4$-linear, and for the binary image of a code over $\Z_4$ 
to be
a linear code. In \cite{wolfmann2} and \cite{wolfmann}, Wolfmann studies the 
condition for the image of a cyclic code over $\Z_4$ of odd length to be 
linear. 
Moreover, he
proves that the cyclic structure is preserved after a convenient permutation
of coordinates.

The aim of this paper is to determine all $\add$-additive cyclic codes with odd 
$\beta$ whose Gray images are linear binary codes. Furthermore, the binary 
generators of the linear image of a $\add$-additive cyclic code are determined 
under the Nechaev-Gray map.

\section{Preliminaries}
\subsection{$\add$-additives codes}

A ${\mathbb{Z}}_2{\mathbb{Z}}_4$-additive code ${\cal C}$ is a subgroup of
${\mathbb{Z}}_2^\alpha\times{\mathbb{Z}}_4^\beta$ (see \cite{AddDual}). For a 
vector ${\bf u} \in \Z_2^\alpha\times\Z_4^\beta$, we write ${\bf u}=(u\mid 
u')$, 
where $u=(u_0,\dots,u_{\alpha-1})\in\Z_2^\alpha$ and 
$u'=(u'_0,\dots,u'_{\beta-1})\in\Z_4^\beta$. Since ${\cal C}$ is a
subgroup of $\mathbb{Z}_2^\alpha\times\mathbb{Z}_4^\beta$, it is also isomorphic
to a commutative structure of the form 
$\mathbb{Z}_2^\gamma\times\mathbb{Z}_4^\delta$
and it has $|{\cal C}| = 2^{\gamma +2\delta}$ codewords.

Let $X$ (respectively $Y$) be the set of $\mathbb{Z}_2$ (respectively
$\mathbb{Z}_4$) coordinate positions, so $|X| =\alpha$ and $|Y| = \beta$. Unless
otherwise stated, the set $X$ corresponds to the first $\alpha$ coordinates and
$Y$ corresponds to the last $\beta$ coordinates. Let
$\C_X $
be the binary punctured code of ${\cal C}$ by deleting the
coordinates outside $X$. Define similarly the quaternary code $\C_Y$.

Let ${\cal C}_b$ be the subcode of ${\cal C}$ which contains all order two 
codewords and let $\kappa$ be the
dimension of $({\cal C}_b)_X$, which is a binary linear code. For the case 
$\alpha = 0$, we write $\kappa = 0$. With all these parametres, we say that the 
code ${\cal C}$ is of type $(\alpha, \beta; \gamma, \delta; \kappa).$ Moreover, 
in \cite{AddDual}, it is shown that $\C$ is permutation
equivalent to a
$\add$-additive code with standard generator
matrix of the form:
\begin{equation}\label{eq:StandardForm}
\cG_S= \left ( \begin{array}{cc|ccc}
I_{\kappa} & T_b & 2T_2 & \zero & \zero\\
\zero & \zero & 2T_1 & 2I_{\gamma-\kappa} & \zero\\
 \zero & S_b & S_q & R & I_{\delta} \end{array} \right ),
\end{equation} \noindent where $I_k$ is the identity matrix of size $k\times 
k$; 
$T_b, S_b$
are matrices over $\Z_2$;  $T_1, T_2, R$ are
matrices over $\Z_4$ with all entries in $\{0,1\}\subset\Z_4$;
and $S_q$ is a matrix over $\Z_4$.

A $\add$-additive code $\C$ is said to be separable if $\C = \C_X \times \C_Y$.
Otherwise the code is said to be non-separable.

Let $u'=(u'_0,\dots, u'_{n-1})$ be an element of $\mathbb{Z}_4^n$ such that 
$u'_i=\tilde{u}'_i + 2 \hat{u}'_i$, for $i=0, \dots, n-1$ and with 
$\tilde{u}'_i, \hat{u}'_i\in \{0,1\}.$ As in \cite{sole}, the \emph{Gray map} 
$\phi$ from $\mathbb{Z}_4^n$ to $\mathbb{Z}_2^{2n}$ is defined by
$$\phi({ u}')= (\hat{u}'_0,\dots, \hat{u}'_{n-1}, \tilde{u}'_0 + 
\hat{u}'_0,\dots, \tilde{u}'_{n-1}+\hat{u}'_{n-1} ).$$

The \emph{Nechaev permutation} is the permutation $\sigma$ on 
$\mathbb{Z}_2^{2n}$ with $n$ odd defined by
$$\sigma(v_0,v_1,\dots,v_{2n-1}) = (v_{\tau(0)}, 
v_{\tau(1)},\dots,v_{\tau(2n-1)}),$$
where $\tau$ is the permutation on $\{0,1,\dots, 2n-1\}$ given by
$$(1,n+1) (3,n+3)\cdots (2i+1, n+2i+1)\cdots (n-2, 2n-2).$$
Let $\psi$ be the map from $\mathbb{Z}_4^n$ into $\mathbb{Z}_2^{2n}$ defined by 
$\psi = \sigma\phi$, with $n$ odd. The map $\psi$ is called the 
\emph{Nechaev-Gray map}, \cite{wolfmann}.

The \emph{extended Gray map} $\Phi$ and the \emph{extended Nechaev-Gray map} 
$\Psi$ are the maps from $\mathbb{Z}_2^\alpha\times\mathbb{Z}_4^\beta$ to 
$\mathbb{Z}_2^{\alpha+2\beta}$ given by
$$\Phi({u}\mid {u'})= ({ u}\mid \phi({u'})),\quad\Psi({ u}\mid{ u'})= ({ u}\mid 
\psi({ u'})).$$

\subsection{$\add$-additive cyclic codes}

Cyclic codes have been a primary area of study for coding theory, 
\cite{macwilliams}. Recently, the class of $\add$-additive cyclic codes has 
been 
defined in \cite{Abu}.

For ${\bf u} = (u \mid u' ) \in \Z_2^\alpha \times\Z_4^\beta$, define the 
cyclic 
shift $\pi$ by
$\pi({\bf u}) = (\pi (u) \mid \pi(u')) $, where $\pi (u) = \pi ( 
u_0,u_1,\dots,u_{\alpha-1}) = (u_{\alpha-1},u_0,u_1,\dots,u_{\alpha-2})$ and 
$\pi (u') = (u'_{\beta-1},u'_0,u'_1,\dots,u'_{\beta-2})$ .
We say that a $\add$-additive code $\C$ is cyclic if $\pi(\C) = \C.$

There exists a bijection between $\Z_2^\alpha \times\Z_4^\beta$ and 
$R_{\alpha,\beta}=\mathbb{Z}_2[x]/(x^\alpha-1)\times\mathbb{Z}_4[x]
/(x^\beta-1)$ 
given by:
\begin{align*}
(u_0, u_1,\dots, u_{\alpha-1}\mid  u'_0,&\dots, u'_{\beta-1})\mapsto\\ &(u_0+ 
u_1x+\dots +  u_{\alpha-1}x^{\alpha-1}\mid  u'_0+\dots 
+u'_{\beta-1}x^{\beta-1}).
\end{align*}
Therefore, as it is common in the studies of cyclic codes, any codeword is 
identified as a vector or as a polynomial.

From now on, the binary reduction of a polynomial $p(x) \in \mathbb{Z}_4[x]$ 
will be denoted by $\tilde{p}(x) $.
Let $p (x) \in \Z_4 [x]$ and $ (b(x)\mid a(x)) \in R_{\alpha, \beta}$, and 
consider the following multiplication $p(x)\star(b(x)\mid 
a(x))=(\tilde{p}(x)b(x)\mid p(x)a(x)).$ From \cite{Abu}, $R_{\alpha, \beta}$ is 
a $\Z_4 [x]$-module with respect to this multiplication.

Let ${ u}'(x) =\tilde{{ u}}'(x) + 2\hat{u}'(x) $ be the polynomial 
representation of ${ u}'\in\mathbb{Z}_4^n$.
Then, the polynomial version of the Gray map is $\phi({ u}'(x) )= (\hat{u}'(x) 
, 
\tilde{ u}'(x) +
\hat{u}'(x) ).$ In the following, a polynomial $p (x) \in \Z_2 [x]$ or $\Z_4 
[x]$ will be denoted simply by $p $.

Using the polynomial representation, an equivalent definition of 
$\add$-additive 
cyclic codes is the following.

\begin{definition}[\cite{Abu}]
A subset $\C\subseteq R_{\alpha,\beta}$ is called a $\add$-additive cyclic code 
if $\C$ is a $\Z_4[x]$-submodule of $R_{\alpha,\beta}$.
\end{definition}

From \cite{Abu}, if $\beta$ is odd, we know that if $\C$ is a $\add$-additive 
cyclic code then it is of the form
\begin{equation} \label{form}
\langle (b\mid { 0}), (\ell \mid fh +2f) \rangle,
\end{equation}
where $fhg = x^\beta -1$ in $\Z_4[x]$, and $b$ divides $x^\alpha-1$ in 
$\Z_2[x]$, and we can assume that $deg(\ell) < deg(b).$
In this case we have that $|C| = 2^{\alpha - deg(b)} 4^{deg(g) } 2^{deg(h)}.$
From now on we consider that $\beta$ is odd. Then $f$, $g$ and $h$ are pairwise 
coprime
polynomials. Since $h$ and $g$ are
coprime, there exist polynomials $\lambda$ and $\mu$, that will be used later
along the paper, such that
\begin{equation}\label{eq:lambda-mu}
 \lambda h+\mu g=1.
\end{equation}

\begin{lemma}[{\cite[Corollary 2]{Z2Z4CDual}}]\label{bdiviXSgcd}
Let ${\cal C}$ be a ${\mathbb{Z}_2 {\mathbb{Z}_4}}$-additive cyclic code of 
type 
$(\alpha, \beta; \gamma , \delta; \kappa)$ with ${\cal C} = \langle (b \mid { 
0}), (\ell  \mid  f h  +2f ) \rangle$. Then, $b $ divides $\frac{x^\beta 
-1}{\tilde{f} } \gcd(b ,\ell )$ and $b $ divides $\tilde{h}  \gcd(b ,\ell 
\tilde{g} ).$
\end{lemma}

We can put the generator matrix (\ref{eq:StandardForm}) in the following form, 
\cite{Z2Z4CDual}.
\begin{equation} \label{gen}
\left(\begin{array}{ccc|ccc}
I_{\kappa_1} & T& T_{b_1}& {\bf 0}&{\bf 0}&{\bf 0} \\
{\bf 0} & I_{\kappa_2} & T_{b_2} & 2T_2 & {\bf 0} & {\bf 0}  \\
{\bf 0} & {\bf 0} & {\bf 0} & 2T_1 & 2 I_{\gamma - (\kappa_1 + \kappa_2)}& {\bf 
0} \\
{\bf 0} & {\bf 0} & S' & S  & R& I_{\delta}
\end{array}\right).
\end{equation}

\begin{theorem}[{\cite[Theorem 5 and Proposition 
6]{Z2Z4CDual}}]\label{TypeDependingDeg}
Let $\C$ be a $\add$-additive cyclic code of type $(\alpha, \beta; \gamma , 
\delta; \kappa)$ with $\C = \langle (b\mid { 0}), (\ell \mid fh +2f) \rangle ,$ 
where $fhg = x^\beta -1.$
Then
\begin{align*}
\gamma &= \alpha -\deg(b)+\deg(h),\\
\delta &= \deg(g),\\
\kappa &= \alpha -\deg(\gcd(\ell \tilde{g}, b))
\end{align*}
and
$$\kappa_1= \alpha -\deg(b),\quad \kappa_2= \deg(b)-\deg(\gcd(b, \ell 
\tilde{g})),$$
$$\delta_1= \deg(\gcd(b,\ell \tilde{g})) - \deg(\gcd(b,\ell)) \mbox{ and }
\delta_2=\deg(g)-\delta_1.$$
\end{theorem}

\begin{proposition}\label{2GenOrderTwoCode}
Let $\C$ be a $\add$-additive cyclic code with $\C = \langle (b\mid{0}), (\ell
\mid fh +2f) \rangle$. Then
$\C_b=\langle (b\mid{ 0}), (\bar{\mu}\ell \tilde{g} \mid 2f)\rangle.$
\end{proposition}
\begin{proof}
By {\cite[Lemma 3]{Z2Z4CDual}}, $\C_b=\langle (b\mid{ 0}), (\ell \tilde{g} \mid
2fg), ({ 0}\mid 2fh) \rangle$ with $\gcd(h, g)=1$.
Note that $(b\mid 0)\in \C$ and $g\star(\ell \mid f h  +2f )=(\ell \tilde{g} 
\mid 2fg)\in\C$. Since $(0\mid 2fh)\in \C$, we obtain $\mu\star(\ell 
\tilde{g}\mid 2fg)+\lambda\star(0\mid 2fh)= (\tilde{\mu}\ell \tilde{g}\mid 
2f)\in\C$ for $\lambda, \mu$ in (\ref{eq:lambda-mu}). Hence, 
$(\tilde{\mu}\ell\tilde{g}\mid 2f)\in \C_b$.

Let $\C'_b=\langle (b\mid{ 0}), (\tilde{\mu}\ell \tilde{g} \mid 2f)\rangle$. By
construction, $\C'_b\subseteq \C_b$.
We shall prove that $\C_b\subseteq\C'_b$. It is an evidence that $(b\mid
0)\in\C'_b$. Since $\lambda h+\mu g=1$ we have that $\tilde{\lambda}
\tilde{h}\ell\tilde{g}+\tilde{\mu} \tilde{g}\ell\tilde{g}=\ell\tilde{g}$. By
Lemma \ref{bdiviXSgcd}, we know that $b$ divides $\bar{h}\ell\bar{g}$. So
$\bar{\mu} \tilde{g}\ell\tilde{g}=\ell\tilde{g}\pmod{b}$. Therefore
$g\star(\tilde{\mu}\ell \tilde{g} \mid 2f)= (\tilde{g}\tilde{\mu}\ell \tilde{g}
\mid 2fg)=(\ell \tilde{g} + p b\mid 2fh)=(\ell \tilde{g}\mid 2fh)+(pb\mid 0)
\in \C'_b$, for some $p\in\Z_2[x]$. Hence $(\ell \tilde{g}\mid 2fh)
\in \C'_b$.

Applying Lemma \ref{bdiviXSgcd}, $h\star(\tilde{\mu}\ell \tilde{g} \mid 2f)=
(\tilde{h}\tilde{\mu}\ell \tilde{g} \mid 2fh)=(p'b\mid 2fh)\in \C'_b$. Thus, 
$(0\mid 2fh)\in \C'_b$, and hence
$\C_b\subseteq \C'_b$.
\end{proof}

Note that if $\C$ is a $\add$-additive cyclic code then ${\cal C}_X$ and ${\cal
C}_Y$ are a cyclic code over $\mathbb{Z}_2$ and a cyclic code over
$\mathbb{Z}_4$ generated
by $gcd(b,\ell)$ and $(fh+2f)$, respectively (see
\cite{macwilliams}, \cite{Wan}).
The following example illustrates that the converse is not
true in general.
\begin{example}
Let $\C$ be a $\add$-additive code generated by
\begin{equation} \label{ex:CxCyCyclicCno}
\left(\begin{array}{cc|ccc}
1& 0& 1& 0&{ 0} \\
0 & 1 & 0 & { 1} & {0}  \\
{ 0} & {0} & 0 & 0& {1}
\end{array}\right).
\end{equation}
Clearly, $\C_X$ and $\C_Y$ are cyclic codes. But $\pi(0,0\mid 0,0,1)=(0,0\mid 
1,0,0)$ does not belong to $\C$.

\end{example}

Nevertheless, the next theorem shows that if the $\add$-additive cyclic codes 
is 
separable then the converse is satisfied.
\begin{theorem}
A $\Z_2 \Z_4$-additive separable  code is cyclic if and only if $\C_X$ is a
binary cyclic code and $\C_Y$ is a $\Z_4$-additive cyclic code. Moreover,
$\C=\langle (b\mid { 0}), (0 \mid fh +2f)
\rangle$.
\end{theorem}
\begin{proof}
Let $\C$ be a separable $\Z_2\Z_4$-additive cyclic code.  Let $u \in \C_X $ and 
$u' \in \C_Y.$
Then $(u\mid {\bf 0 }) \in \C$ and $\pi(u\mid {\bf 0 })
= (\pi(u)\mid {\bf 0}) \in \C$ which gives that $\pi(u) \in \C_X$.  Therefore, 
$C_X$
is cyclic.
Similarly, $({\bf 0 }\mid u') \in \C$ and $\pi(({\bf 0 }\mid u')) = ({\bf 0 
}\mid \pi(u'))
\in \C$ which gives that $\pi(u') \in \C_Y$. Therefore, $C_Y$ is cyclic.

If both $C_X$ and $C_Y$ are cyclic then $\pi(u\mid u') = (\pi(u)\mid \pi(u')) 
\in \C_X
\times \C_Y = \C$ and so $\C$ is cyclic. The polynomial representation follows 
immediately from the standard form of the generator polynomials given in 
(\ref{form}) since the code is separable.
\end{proof}

\section{Binary images of $\add$-additive cyclic codes}

It is well known that for a $\add$-additive code $\C$, the binary code 
$\Phi(\C)$
may not be linear, \cite{Z2Z4RK}.

Let ${\bf u} = (u \mid u' )$ and $\vv = (v \mid v')$ be in $\Z_2^\alpha
\times\Z_4^\beta$, and denote the componentwise product as ${\bf u}*\vv =
(u*v\mid u'*v')$.
From \cite{Z2Z4RK}, we have that
\begin{equation} \label{equation:linear}
\Phi(\vv+\vw) =
\Phi(\vv) + \Phi(\vw) + \Phi(2\vv*\vw).
\end{equation}
It follows immediately that $\Phi(\C)$ is
linear if and only if $2 {\bf u} * \vv \in \C$, for all ${\bf u}, \vv\in\C$.

Let ${p} $ be a divisor of $x^n-1$ in $\mathbb{Z}_2[x]$ with $n$ odd and let 
$\xi$ be a primitive $n$th root of unity over $\mathbb{Z}_2$. The polynomial 
$({p}\otimes{p}) $ is defined as the divisor of $x^n-1$ in $\mathbb{Z}_2[x]$ 
whose roots are the products $\xi^i\xi^j$ such that $\xi^i$ and $\xi^j$ are 
roots of ${p} .$

\subsection{Images under the Gray map}
In \cite{wolfmann}, the author characterizes all linear cyclic codes over 
$\Z_4$ 
of odd length whose Gray map images are linear binary codes as it is shown in 
the following theorem.
\begin{theorem}[{\cite[Theorem 20]{wolfmann}}]\label{Phi(CY)Linear}
Let ${\cal D}=\langle f h  +2f  \rangle$ be a ${{\mathbb{Z}_4}}$-additive 
cyclic 
code of odd length $n$, and where $f h g  = x^n -1.$ The following properties 
are equivalent.
\begin{enumerate}
\item $\gcd(\tilde{f} , (\tilde{g}\otimes\tilde{g}) )=1$ in $\mathbb{Z}_2[x]$;
\item $\phi ({\cal D})$ is a binary linear code of length $2n$.
\end{enumerate}
\end{theorem}

Our aim in this section is to give a classification of all $\add$-additive 
cyclic codes with odd $\beta$ whose Gray images are linear binary codes. First, 
we will show that Theorem \ref{Phi(CY)Linear} is a necessary but not a 
sufficient condition for $\add$-additive cyclic codes.

\begin{lemma}
Let $\C$ be a $\add$-additive code such that $\Phi(\C)$ is linear. Then, $\C_Y$ 
is a $\Z_4$-additive code such that $\phi(\C_Y)$ is linear.
\end{lemma}

Note that the converse is not true in general. The following example 
illustrates 
it.
\begin{example}\label{example}
Consider the $\add$-additive code generated by the following matrix.
\begin{equation}
\left(\begin{array}{ccc|ccc}
1&0&0&0&0&0 \\
0&1&0&0&0&0\\
0&0&1&2&0&0\\
0&0&0&1&1&0\\
0&0&0&1&0&1\\
\end{array}\right)
 \end{equation}

Clearly, $\phi(\C_Y)$ is binary linear since $2(1,1,0)*(1,0,1) =(2,0,0) \in 
\C_Y$. Nevertheless, $\Phi(\C)$ is not binary linear because 
$2(0,0,0|1,1,0)*(0,0,0|1,0,1) = (0,0,0|2,0,0)$ that does not belong to $\C$.
\end{example}

Before we will be able to determine the classification of $\add$-additive 
cyclic 
codes, we must introduce a couple of lemmas.

\begin{lemma}\label{CbSeparable}
Let ${\cal C}=\langle (b \mid 0), (\ell \mid f h  +2f ) \rangle$ be a 
$\add$-additive cyclic code of length $\alpha+\beta$, $\beta$ odd, and where $f 
h g  = x^\beta -1.$ If $\C_b=\langle (b\mid 0),(0\mid 2f)\rangle$, then 
$\Phi(\C)$ is linear if and only if $\phi(\C_Y)$ is linear.
\end{lemma}
\begin{proof}
Assume $\Phi(\C)$ is linear. Let $u', v'\in\C_Y$. There exist 
$\textbf{u}=(u\mid 
u'), \textbf{v}=(v\mid v')\in\C$. Since $\Phi(\C)$ is linear $2\textbf{u} \ast 
\textbf{v}=({\bf 0 }\mid 2u'\ast v')\in\C$ and hence $2u'\ast v'\in\C_Y$  and 
so 
$\phi(\C_Y)$ is linear.

In the case that $\phi(\C_Y)$ is linear, we consider $\textbf{u},\textbf{v}\in 
\C$. We have that $2\textbf{u}\ast \textbf{v}=({\bf 0 }\mid 2u'\ast v').$ Since 
$2u'\ast v'\in\C_Y$ and $\C_b=\langle (b\mid 0),(0\mid 2f)\rangle$, then $({\bf 
0 }\mid 2u'\ast v')\in \C$ and $\Phi(\C)$ is linear.
\end{proof}

Note that the assumption of Lemma \ref{CbSeparable} is equivalent to say that
$b$ divides $\ell \tilde{g}$ for a code ${\cal C}=\langle (b\mid 0), (\ell\mid
fh +2f) \rangle$. Theorem \ref{Phi(CY)Linear} shows when a $\Z_4$-cyclic code
${\cal D}$ has linear binary image by $\phi$. So our next objective is to give
an analogous result for $\add$-additive cyclic codes.

In general, a $\add$-additive code has linear binary image if for vectors ${\bf 
v}$ and ${\bf w}$ of the last $\delta$ rows of (\ref{gen}) then $2{\bf 
v}\ast{\bf w}\in\C$, \cite{Z2Z4RK}.

Let $\C$ be a $\add$-additive cyclic code with generator matrix as in 
(\ref{gen}), and denote by $\C'$ the subcode generated by the last $\delta$ 
rows 
of (\ref{gen}). Note that $\C'$ may not be a $\add$-additive cyclic code and 
that $\C'$ has a generator matrix of the form
\begin{equation} \label{genSubcode}
\left(\begin{array}{ccc|ccc}
{\bf 0} & {\bf 0} & {\bf 0} & 2T_2 & 2 I_{\gamma - (\kappa_1 + \kappa_2)}& {\bf
0} \\
{\bf 0} & {\bf 0} & S' & S  & R& I_{\delta}
\end{array}\right).
\end{equation}

\begin{lemma}
Let $\C$ be a $\add$-additive code with generator matrix (\ref{gen}). Then, 
$\C$ 
has linear binary image if and only if the subcode $\C'$ generated by matrix 
(\ref{genSubcode}) has linear binary image.
\end{lemma}
\begin{proof}
Let $\C'$ the subcode generated by the matrix (\ref{genSubcode}). Note that all 
the rows of the matrix (\ref{gen}) which are not in the matrix 
(\ref{genSubcode}) have order two. Clearly, if $\Phi(\C')$ is linear, then 
$\Phi(\C)$ is linear.

Suppose that $\Phi (\C)$ is linear. For all $\textbf{v}, \textbf{w}\in\C'$ we
have that $2\textbf{v}\ast\textbf{w}\in \C$. Let $\textbf{v}, \textbf{w}\in\C'$.
If $\textbf{v}$ or $\textbf{w}$ has order two then
$2\textbf{v}\ast\textbf{w}={\bf 0 }\in \C'$. If $\textbf{v}$ and $\textbf{w}$ 
have
order four then $2\textbf{v}\ast\textbf{w}$ belongs to the code $\C'_b$
generated by
$$\left(\begin{array}{ccc|ccc}
{\bf 0} & {\bf 0} & {\bf 0} & 2T_1 & 2 I_{\gamma - \kappa}& {\bf 0} \\
{\bf 0} & {\bf 0} & {\bf 0} & 2S  & 2R& 2I_{\delta}
\end{array}\right).
$$
Since $\C_b'\subseteq \C'$, we have that $2\textbf{v}\ast\textbf{w}\in\C'$.
\end{proof}

\begin{lemma}\label{lemma:3generatorsC}
Let ${\cal C}=\langle (b \mid 0), (\ell \mid f h  +2f ) \rangle$ be a 
$\add$-additive cyclic code of length $\alpha+\beta$, $\beta$ odd, and where $f 
h g  = x^\beta -1.$ There exist $\ell'\in \frac{\Z_2[x]}{\langle 
x^\alpha-1\rangle}$ such that $\C$ can be also generated by $(b\mid 0),(\ell 
\tilde{g} \mid 2fg), (\ell'\mid fh)$.
\end{lemma}
\begin{proof}
Consider $\lambda, \mu$ in (\ref{eq:lambda-mu}). Let $\ell'=\ell -\tilde{\mu} 
\ell \tilde{g}$ and ${\cal D}=\langle (b\mid 0),(\ell \tilde{g} \mid 2fg), 
(\ell'\mid fh)\rangle$. We shall prove that ${\cal D}=\C$.

By Proposition~\ref{2GenOrderTwoCode}, $(\tilde{\mu}\ell \tilde{g}\mid 
2f)\in\C$. Hence, $(\ell \mid f h  +2f )-(\tilde{\mu} \ell \tilde{g}\mid 
2f)=(\ell'\mid fh)\in\C$. Therefore, ${\cal D}\subseteq\C$.
Finally, since $(0\mid 2fh)\in {\cal D}$ and $(\ell\tilde{g}\mid 2fg)\in {\cal 
D}$, we have that $(\tilde{\mu}\ell\tilde{g}\mid 2f)\in {\cal D}.$ Therefore 
$(\ell\mid fh+2f)= (\ell'\mid fh) + (\tilde{\mu}\ell\tilde{g}\mid 2f)\in {\cal 
D}$, that implies $\C\subseteq {\cal D}$.
\end{proof}

\begin{theorem}\label{Phi(C)Linear}
Let ${\cal C}=\langle (b \mid 0), (\ell \mid f h  +2f ) \rangle$ be a 
$\add$-additive cyclic code of length $\alpha+\beta$, $\beta$ odd, and where $f 
h g  = x^\beta -1.$ The following properties are equivalent.
\begin{enumerate}
\item $\gcd(\frac{\tilde{f}b}{\gcd(b, \ell \tilde{g})}, 
(\tilde{g}\otimes\tilde{g}) )=1$ in $\mathbb{Z}_2[x]$;
\item $\Phi ({\cal C})$ is a binary linear code of length $\alpha + 2\beta$.
\end{enumerate}
\end{theorem}
\begin{proof}
By Lemma \ref{lemma:3generatorsC}, ${\cal C}=\langle (b\mid 0),(\ell \tilde{g} 
\mid 2fg), (\ell'\mid fh)\rangle$, for $\ell'= \ell-{\mu}\ell{g}$. Clearly, the 
subcode $\langle (b\mid{ 0}), (\ell {g} \mid 2fg), ({0}\mid 2fh) \rangle$ has 
binary linear image. So, the question is when $\langle (\ell'\mid fh)\rangle$ 
has linear image.

Since $ (b\mid{ 0})$ and $(\ell \tilde{g} \mid 2fg)$ belong to $\C$, we have 
$$\left(0\mid \frac{2fgb}{\gcd(b, \ell \tilde{g})}\right)=\frac{\ell g}{\gcd(b, 
\ell \tilde{g})}\star(b\mid 0)+ \frac{b}{\gcd(b, \ell \tilde{g})}\star(\ell 
\tilde{g}\mid 2fg)\in\C.$$

Consider the subcode $\C'$ generated by $(\ell'\mid fh)$ and $\left(0\mid 
\frac{2fgb}{\gcd(b, \ell \tilde{g})}\right)$. Since $g$ and $h$ are coprime 
over 
$\Z_4$, then $\C'=\langle \left(\ell'\mid fh + \frac{2fb}{\gcd(b, \ell 
\tilde{g})}\right)\rangle$. Note that $\Phi(\C)$ is linear if and only if 
$\Phi(\C')$ is linear.

By Lemma \ref{bdiviXSgcd}, we can consider the polynomials 
$f'=\frac{fb}{\gcd(b, 
\ell \tilde{g})}$, $h'=\frac{h\gcd(b, \ell \tilde{g})}{b}$ and $g'=g$. It is 
clear that $f'g'h'=x^\beta-1$, with $f', h'$ and $g'$ coprime factors. Then we 
can write $\C'=\langle (\ell'\mid h'f'+2f')\rangle$. By Lemma 
\ref{CbSeparable}, 
$\Phi(\C')$ is a binary linear code if and only if $\phi(\C'_Y)$ is linear. By 
Theorem \ref{Phi(CY)Linear}, $\phi(\C'_Y)=\phi(\langle (f'h'+2f')\rangle)$ is 
linear if and only if $\gcd(\tilde{f'}, (\tilde{g}'\otimes\tilde{g}'))=1$, this 
is equivalent to say that $\gcd(\frac{\tilde{f}b}{\gcd(b, \ell \tilde{g})}, 
(\tilde{g}\otimes\tilde{g}))=1$.
\end{proof}

The following result describes an infinite family of $\add$-additive cyclic 
codes with odd $\beta$ whose Gray images are linear.

\begin{theorem}
Let ${\cal C}=\langle (b \mid 0), (\ell \mid f h  +2f ) \rangle$ be a 
$\add$-additive cyclic code of length $\alpha+\beta$, $\beta$ odd, and where $f 
h g  = x^\beta -1.$ If $g =1$ or $g =x^s-1$ where $s$ divides $\beta$, then 
$\Phi({\cal C})$ is linear.
\end{theorem}
\begin{proof}
Clearly, $(\tilde{g}\otimes\tilde{g}) =\tilde{g} $ since the set of roots of 
$\tilde{g} $ is a multiplicative subgroup of the splitting field of 
$x^\beta-1$. 
Then, $\gcd(\frac{\tilde{f}b}{\gcd(b, \ell \tilde{g})}, 
(\tilde{g}\otimes\tilde{g}) )=1$. By Theorem \ref{Phi(C)Linear}, $\Phi(\C)$ is 
linear.
\end{proof}

In \cite{Z2Z4RK}, the authors show that there exist a $\add$-additive code for 
a 
given type $(\alpha, \beta; \gamma,\delta;
\kappa)$ which image is a linear binary code. Considering 
Theorem~\ref{Phi(C)Linear}, the next example illustrates that this
result is not true for $\add$-additive cyclic codes; i.e., for a given type
$(\alpha, \beta; \gamma,\delta; \kappa)$ there does not always exist a
$\add$-additive cyclic code $\C$ with $\Phi(\C)$ linear.

\begin{example}\label{NotAllTypeLinear}
Let $\C$ be a $\add$-additive code of type $(2,7; 2,3; \kappa)$. We will see 
that there does
not exist a $\add$-additive cyclic code of type $(2,7; 2,3;\kappa)$, $\C$, with 
$\Phi(\C)$ a linear binary code for any possible value of $\kappa$.

Let $\alpha=2$ and $\beta=7$. It is well known that $x^2-1= (x-1)^2$ over
$\Z_2$, and $x^7-1=(x-1)(x^3 + 2x^2 + x + 3)(x^3 + 3x^2 + 2x + 3)$.
Let ${\cal C}=\langle (b\mid 0), (\ell\mid fh +2f) \rangle$ be a
$\add$-additive cyclic code of type $(2,7; 2,3; \kappa)$, where $fhg = x^7 -1$.

By Theorem~\ref{TypeDependingDeg}, $\deg(g)=3$ and $\deg(b)=\deg(h)\leq 2$. Let
$\{p_3,q_3\}=\{(x^3 + 2x^2 + x + 3),(x^3 + 3x^2 + 2x + 3)\}$. Assume, without
loss of generality, that $g=p_3$ and, since $\deg(h)\leq 2$, we have that $q_3$
divides $f$. It is easy to see that
$\gcd(q_3,(\tilde{p_3}\otimes\tilde{p_3}))\not=1$ and therefore
$\gcd(\frac{\tilde{f}b}{\gcd(b, \ell \tilde{g})}, (\tilde{g}\otimes\tilde{g})
)\neq 1$. Hence, by Theorem~\ref{Phi(C)Linear}, there does not exist a
$\add$-additive code of type $(2,7; 2,3; \kappa)$ with linear Gray image, for 
any possible value of $\kappa$.
\end{example}

\subsection{Images under the Neachaev-Gray map}

In this section, we show that the convenient map to obtain a cyclic structure 
on 
the binary images from $\add$-additive cyclic codes is considering the 
Nechaev-Gray map instead of the Gray map. First, we have to introduce the 
double 
cyclic property on binary codes described in \cite{Z2double} (this family is a 
subfamily of generalized quasi-cyclic codes, see \cite{GenQuaCy}).

Let $r$ and $s$ be positive integers. A binary linear code $C$ of length $r+s$ 
is a $\mathbb{Z}_2$-double cyclic code if the set of coordinates can be 
partitioned into two subsets, the first $r$ and the last $s$ coordinates, such 
that any cyclic shift of the coordinates of both subsets leaves invariant the 
code. These codes can be identified as submodules of the 
$\mathbb{Z}_2[x]$-module $\mathbb{Z}_2[x]/\langle 
x^r-1\rangle\times\mathbb{Z}_2[x]/\langle x^s-1\rangle$, \cite{Z2double}.

\begin{theorem}[\cite{Z2double}]\label{Z2DoubleStructure}
Let $C$ be a $\mathbb{Z}_2$-double cyclic code of length $r+s$. Then $C$ is 
genereted by $$\langle(b \mid { 0}),(\ell \mid  a 
)\rangle\subseteq\frac{\mathbb{Z}_2[x]}{\langle 
x^r-1\rangle}\times\frac{\mathbb{Z}_2[x]}{\langle x^s-1\rangle},$$
where $a|(x^s-1)$, $b|(x^r-1)$ and we can assume that $\deg (\ell)<\deg(b)$.
\end{theorem}

From \cite{Z2double}, we obtain the following useful result.

\begin{theorem}[{\cite[Theorem 6]{Z2double}}]\label{PsiCZ2DoubleCy}
Let $\C$ be a $\add$-additive cyclic code. If $\Phi(\C)$ is a linear binary 
code 
then $\Psi(\C)$ is a $\Z_2$-double cyclic code.
\end{theorem}

If $\C$ is a  $\add$-additive cyclic code which has binary linear image then 
the 
following theorem shows the relation established between the generator 
polynomials of $\C$ and the generators of its $\Z_2$-double cyclic image.

\begin{theorem}
Let ${\cal C}=\langle (b\mid { 0}), (\ell\mid  fh +2f) \rangle\subseteq 
\mathbb{Z}_2^\alpha\times\mathbb{Z}_4^\beta$ be a ${\mathbb{Z}_2 
{\mathbb{Z}_4}}$-additive cyclic code, where $\beta$ is an odd integer and $fhg 
= x^\beta -1.$ Let $\Psi$ be the extended Nechaev-Gray map. If 
$\gcd(\frac{\tilde{f}b}{\gcd(b, \ell \tilde{g})}, 
(\tilde{g}\otimes\tilde{g}))=1$ in $\mathbb{Z}_2[x]$, then $\Psi ({\cal C})$ is 
a $\mathbb{Z}_2$-double cyclic code of length $\alpha+2\beta$ generated by
$$\Psi ({\cal C})=\langle (b\mid { 0}), (\ell' \mid  \tilde{f}^2\tilde{h}) 
\rangle,$$
where $\ell'=\tilde{p}\ell\pmod{ b}\in\frac{\Z_2[x]}{\langle 
x^\alpha-1\rangle}$ 
such that $p(fh +2f)=\psi^{-1}(\tilde{f}^2\tilde{h}).$
\end{theorem}

\begin{proof}
Let $\gcd(\frac{\tilde{f} b }{\gcd(b ,\ell  \tilde{g} )}, 
(\tilde{g}\otimes\tilde{g}) )=1$. By Theorem \ref{Phi(C)Linear} and Theorem 
\ref{PsiCZ2DoubleCy}, $\Psi({\cal C}_Y)$ is a $\Z_2$-double cyclic code of 
length $\alpha+2\beta$. By a similar argument of \cite[Theorem 15]{wolfmann}, 
but taking care of the binary part we obtain that $\Psi(\C)=\langle (b\mid 0), 
(\ell'\mid \tilde{f}^2\tilde{h})\rangle$, for some $\ell'\in 
\frac{\Z_2[x]}{\langle x^\alpha-1\rangle}$.
Since $\psi^{-1}(\tilde{f}^2\tilde{h})\in \C_Y=\langle fh+2f\rangle$ there 
exists $p\in\frac{\Z_4[x]}{\langle x^\beta-1\rangle}$ such that 
$\psi^{-1}(\tilde{f}^2\tilde{h})=p(fh+2f)$. Therefore, $(\ell'\mid 
\psi^{-1}(\tilde{f}^2\tilde{h}))= p'\star(b\mid 0)+ p\star(\ell \mid fh +2f)$ 
for some $p'$. Thus, $\ell'=\tilde{p}\ell\pmod{ b}\in \frac{\Z_2[x]}{\langle 
x^\alpha-1\rangle}$.
\end{proof}

\begin{example}
Let $\C=\langle (x^2+x+1\mid 0), (1\mid x^2+x+3)\rangle\subset\Z_2^3\times 
\Z_4^3$ be a $\add$-additive cyclic code, where $f=1$ and $h=x^2+x+1$. Then, we 
have that $\psi^{-1}(\tilde{f}^2\tilde{h})=1+3x+x^2$: 
$$\tilde{f}^2\tilde{h} \longrightarrow (1,1,1,0,0,0) \xrightarrow{\tau^{-1}} 
(0,1,0,1,0,1) \xrightarrow{\phi^{-1}} (1,3,1) \longrightarrow 1+3x+x^2.$$
Therefore, the polynomial $p$ such that $p(fh 
+2f)=\psi^{-1}(\tilde{f}^2\tilde{h})$ is $p=x$ and then

$$ \Psi(\C)=\langle (x^2+x+1\mid 0), (x\mid x^2+x+1)\rangle\subset 
\Z_2^{3}\times\Z_2^6.$$

\end{example}

\section{Conclusions}

In this paper we have studied the binary images of $\Z_2\Z_4 $-additive cyclic 
codes. We determine all $\add$-additive cyclic codes with odd $\beta$ whose 
Gray 
images are linear binary codes and we have given the generators of the images 
under the Nechaev-Gray map as $\Z_2$-double cyclic codes.

\end{document}